\documentclass[onefignum,onetabnum]{siamart171218}

\usepackage{lipsum}
\usepackage{amsfonts}
\usepackage{graphicx}
\usepackage{epstopdf}
\usepackage{algorithmic}
\ifpdf
  \DeclareGraphicsExtensions{.eps,.pdf,.png,.jpg}
\else
  \DeclareGraphicsExtensions{.eps}
\fi

\newsiamremark{remark}{Remark}
\newsiamremark{hypothesis}{Hypothesis}
\crefname{hypothesis}{Hypothesis}{Hypotheses}
\newsiamthm{claim}{Claim}

\headers{Finite population effects on optimal communication}{Hyunjoong Kim, Yoichiro Mori, and Joshua Plotkin}

\title{Finite population effects on optimal communication for social foragers\thanks{Submitted to the editors December xx, 2022}}

\author{Hyunjoong Kim\thanks{Center for Mathematical Biology \& Department of Mathematics, University of Pennsylvania, Philadelphia, PA 19104 (\email{h6kim@sas.upenn.edu})} 
\and Yoichiro Mori\thanks{Center for Mathematical Biology \& Department of Mathematics \& Department of Biology, University of Pennsylvania, Philadelphia, PA 19104 (\email{y1mori@sas.upenn.edu})}
\and Joshua B. Plotkin\thanks{Center for Mathematical Biology \& Department of Mathematics \& Department of Biology, University of Pennsylvania, Philadelphia, PA 19104 (\email{jplotkin@sas.upenn.edu})}}

\usepackage{amsopn}

\makeatletter
\newcommand*{\addFileDependency}[1]{
  \typeout{(#1)}
  \@addtofilelist{#1}
  \IfFileExists{#1}{}{\typeout{No file #1.}}
}
\makeatother

\newcommand*{\myexternaldocument}[1]{%
    \externaldocument{#1}%
    \addFileDependency{#1.tex}%
    \addFileDependency{#1.aux}%
}

\newcommand{\paren}[1]{\left(#1\right)}

\newcommand{\D}[2]{\frac{d#1}{d#2}}
\newcommand{\PD}[2]{\frac{\partial#1}{\partial#2}}

\newcommand{\at}[2]{\left. #1 \right|_{#2}}

\newcommand{\mc}[1]{\mathcal{#1}}
\newcommand{\bm}[1]{\boldsymbol{#1}}
\newcommand{\abs}[1]{\left\lvert #1 \right\rvert}
\newcommand{\norm}[1]{\left\lVert #1 \right\rVert}
\newcommand{\dual}[2]{\left\langle #1,#2 \right\rangle}

\ifpdf
\hypersetup{
  pdftitle={Finite population effects on optimal communication for social foragers},
  pdfauthor={H. Kim, Y. Mori, and J. B. Plotkin}
}
\fi

\myexternaldocument{ex_supplement}

\begin{document}

\maketitle

\begin{abstract}
Foraging is crucial for animals to survive. Many species forage in groups, as individuals communicate to share information about the location of available resources. For example, eusocial foragers, such as honey bees and many ants, recruit members from their central hive or nest to a known foraging site. However, the optimal level of communication and recruitment depends on the overall group size, the distribution of available resources, and the extent of interference between multiple individuals attempting to forage from a site. In this paper, we develop a discrete-time Markov chain model of eusocial foragers, who communicate information with a certain probability. We compare the stochastic model and its corresponding infinite-population limit.
We find that foraging efficiency tapers off when recruitment probability is too high -- a phenomenon that does not occur in the infinite-population model, even though it occurs for any finite population size. 
The marginal inefficiency at high recruitment probability increases as the population increases, similar to a boundary layer. In particular, we prove there is a significant gap between the foraging efficiency of finite and infinite population models in the extreme case of complete communication.  We also analyze this phenomenon by approximating the stationary distribution of foragers over sites in terms of mean escape times from multiple quasi-steady states. We conclude that for any finite group of foragers, an individual who has found a resource should only sometimes recruit others to the same resource. We discuss the relationship between our analysis and multi-agent multi-arm bandit problems.
\end{abstract}



\section{Introduction}

Foraging is a crucial behavior for animals to survive and reproduce. Many species forage in groups, where individuals share information about available resources or possible predators. In such settings, the total group size may have a strong effect on foraging behavior and efficiency. Larger communities are often more successful than smaller ones, due to the benefits of cooperation and information sharing \cite{Boesch2000,Pena2018}. However, the relationship between group size and foraging efficiency is not always straightforward. Some research has suggested that there may be diminishing returns to group size, meaning that the benefits of group living and communication may taper off as group size increases beyond a certain point \cite{Grueter2018,Ioannou2008,Pena2018}. 

Many social foragers, such as most ants and honeybees, have a particular site (nest) where foragers carry resources back to consume and store -- and they are referred to as central place foragers (CPF). 
Such a center can offer safety against predators (compared to foraging areas) \cite{Bell1990, Ward1965, Ward1973, Zahavi1971} and operate as an information center \cite{Bell1990}. 
Resource distributions are often spatially inhomogeneous, forming clusters or patches. Then a forager’s search direction from the center affects whether they will find a resource. 
An important question is how CPFs determine where to be directed from their center. 
CPFs are allocated by various ways of recruitment, both inside and outside the center, by another individual who has already found a resource. For example, honey bees use ``waggle dances" to instruct other honey bees towards a known food source \cite{Haldane1954,Schurch2015}. Many species of ants make a chemical trail from the center to a food source \cite{David2009,Detrain2008}, and also share the information in their colony by sharing the food sample \cite{Lenoir1982}.
In all these cases, overall foraging efficiency of the group depends on the chance that one individual who knows the location of (one) resource site recruits other foragers to the same site.

There is large body of mathematical models for studying the problem of forager allocation. Many studies on CPF allocation assume an infinite population \cite{Biesmeijer2001,Camazine1991,Dukas1998,Seeley1991} and determine the optimal forager allocation in terms of an ``ideal free distribution" \cite{Fretwell1969}. 
However, stochastic models are required to understand  finite size effects; and it not always the case that the  behavior of large-population stochastic models will approach the behavior of the infinite-population limit. 
Apart from CPFs, there have been many stochastic models in optimal foraging:
Individual search processes have been studied based on random walks \cite{ Arehart2022b, Davidson2019,Garg2021,Kilpatrick2021, Viswanathan1996, Viswanathan1999};
Departure time to another foraging site has been analyzed as a renewal process for individuals \cite{Giraldeau2000,Stephens1982} and by a mechanistic drift-diffusion model \cite{Bidari2022};
Game theoretic frameworks \cite{Giraldeau2000,Pena2018,Schaffer1988} (and citations therein) have been proposed for social groups to understand when group membership benefits individuals. 
However, stochastic studies on CPFs are relatively under-explored.

In this paper, we introduce a finite Markov process to understand how the finite population of CPFs allocated over resource patches. We are especially interested in qualitative deviations between the finite-population model and its infinite population limit. First, we define foraging efficiency $\mu$ in terms of the expected value of the long-term reward rate to the entire group. Then, we introduce a recruitment probability $\rho$ that quantifies the degree of communication at the center. We analyze the optimal recruitment probability $\rho^\star$ that maximizes $\mu$, and then we contrast $\mu$ and $\rho^\star$ as a function of group size, $\xi$. 

Our foraging model of CPFs is closely related to the multi-armed bandit (MAB) problem, which is a mathematical model developed to quantify the explore-exploit trade-off \cite{Lai1985, Lattimore2020, Madhushani2021b, Madhushani2021a, Robbins1952, Sutton2018}.
The classical MAB concerns a single agent making a series of choices among multiple arms (or options) and receiving a reward after each subsequent choice. Different types of MAB problems are defined in terms of (i) the number of agents, (ii) the communication constraint among agents, and (iii) the rules for drawing rewards.
In our model, foragers are analogous to agents, the recruitment process corresponds to the communication constraint, and the foraging process is analogous to the reward constraint. The particular MAB problem that maps to our model is a stochastic multi-agent MAB problem for agents with random communication.
One key feature of our model that distinguishes it from a typical MAB problem, however, is that the reward rule of one agent is not independent of the choices made by other agents. That is, when many agents simultaneously choose the same arm, there is interference that reduces the rate of reward on that arm. This added complications has many potential applications in decision-making problems under feedback between decisions and the environment. In addition, our analysis focuses on the long-run time-averaged reward rather than the finite-time net reward, which is usually considered in other MAB problems. (Although we do discuss the finite-time problem in the context of the reward convergence rate.)

The paper is structured as follows: We summarize our main results in Sect. \ref{sect1.1}.
In Sect. \ref{sect2}, we introduce an infinite population model for CPFs as a discrete-time deterministic process. We then analyze its steady-state solutions and investigate the linear stability analysis of the deterministic model.
In Sect. \ref{sect3}, we introduce a stochastic, finite-population model, which converges to the deterministic model as $\xi \to \infty$, and we investigate the convergence of $\mu$ and $\rho^\star$ as $\xi \to \infty$ numerically. Interestingly, we observe a boundary layer of $\mu(\rho)$ near $\rho = 1$ at large $\xi$ -- so that the infinite-population model has qualitatively different behavior than the finite-population model, regardless of how large the (finite) population size. We analyze the time convergence of the population model by considering relaxation time.
In Sect. \ref{sect4}, we analyze the stationary distribution for the stochastic model at $\rho = 1$, which can explain the existence of the boundary layer.

\subsection{Summary of the main results} \label{sect1.1}
We show that the optimal recruitment probabilities of finite and infinite population models are not the same, even as the population size grows large. If $\xi = \infty$, then $\rho^\star = 1$ (and it is unique except in the case of a uniform resource abundance over patches), as shown in Sect. \ref{sect2}. In other words, it will be optimal to share successful experiences with all other foragers in the infinite-population case. However, $\rho^\star < 1$ if $\xi < \infty$. That is, regardless of the size of the finite population, some inefficiency arises when foragers share their successful experience to too many others. To help understand this counter-intuitive result, in  Sect. \ref{sect4} we analyze how high recruitment can lead too many foragers to a single (most abundant) foraging site, which causes inefficiency by not exploiting other sites. Furthermore, we numerically show that $1-\rho^\star = \mathcal{O}(\xi^{-1})$ in Sect. \ref{sect3}.

More generally, we find that $\mu_\xi(\rho)$ has a boundary layer at $\rho = 1$ for sufficiently large $\xi$. This implies that $\mu_\xi(\rho)$ is non-monotonic in $\rho$ and thus have $\rho^\star < 1$ if $\xi < \infty$.
We proved that
\begin{equation}
	\lim_{\xi \to \infty}\mu_\xi(1) < \lim_{\rho \to 1} \mu(\rho),
\end{equation}
which strongly suggests for a boundary layer at $\rho = 1$. The presence of boundary layer can be proven if we show the limit $\lim_{\xi \to \infty} \mu_\xi (\rho) = \mu(\rho)$ when $\rho < 1$. More details and further discussions can be found in Sect. \ref{sect4}.

Our analysis also reveals that the stationary distribution of the finite-population model does not always concentrate in the vicinity of the ideal free distribution that gives the same foraging probabilities for all sites.  Especially at $\rho = 1$, the deterministic limit model has multiple fixed points with the same foraging probabilities for some sites and zero for others. The stationary distribution concentrates on the fixed point where all foragers are in the most resource-abundant site. This appears more clearly when $\rho$ and $\xi$ are large because it is harder to escape the ``quasi-steady state.''  We investigate why that site is preferred in Sect. \ref{sect4}.

\section{Infinite population model} \label{sect2}
We consider the following foraging model, as illustrated in Fig. \ref{fig1}.
\begin{figure}[t!]
\centering
\includegraphics[width=.8\columnwidth]{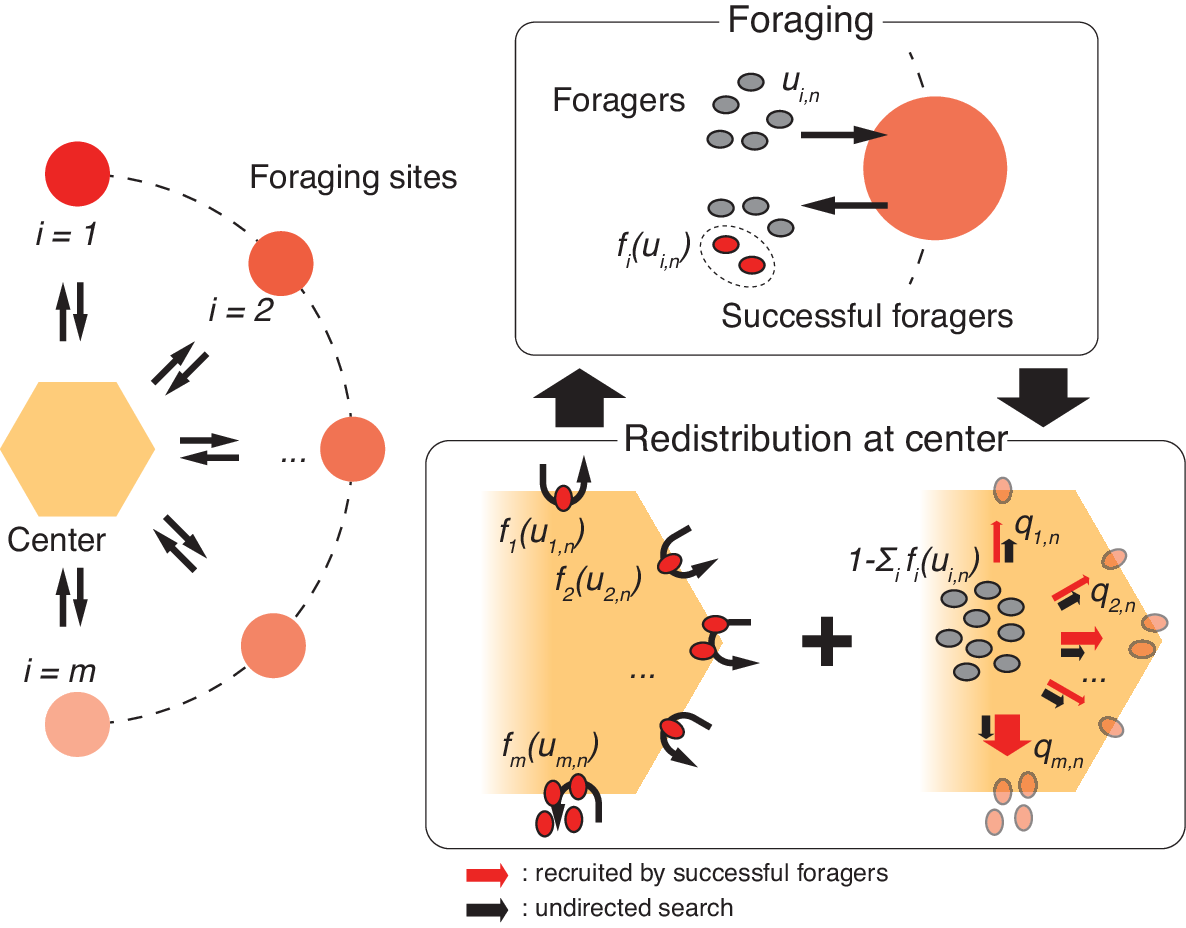}
\caption{\label{fig1} Illustration of foraging model for multiple foraging sites.} 
\end{figure}
There are $m$ foraging sites at a fixed equal distance from the center. 
Foragers fly out to one of the $m$ sites, forage (or attempt to forage), return to the center, 
and fly out again. (We use the term ``fly" generically to represent an attempt to forage from a particular site, analogous pulling a specific ``arm" in the MAB problem.)
We assume that all foragers have fixed speeds and spend negligible time at the foraging site.
Let $u_{i,n}$ be the proportion of foragers that fly to site $i=1,\cdots m$ at time $n=0,1,2,\cdots$. By definition, we have:
\begin{equation}\label{udef}
\sum_{i=1}^m u_{i,n}=1, \; u_{i,n}\geq 0.
\end{equation}
We now introduce a recurrence relation for $u_{i,n}$.
Consider the proportion of foragers that fly out to foraging site $i$ at time $n$. Of those foragers, those that successfully forage are given by:
\begin{equation} \label{success_prob}
v_{i,n}=u_{i,n} \phi_i(u_{i,n}) \equiv f_i(u_{i,n}),
\end{equation}
where $\phi_i(x)$ represents the probability of successfully foraging under competition with $x$ proportion of foragers.
We assume that the $f_i(x)$ satisfies certain structural conditions. 
From the natural assumption $f_i(x)\leq x$, we choose $\phi_i$ (and hence $f_i$)
\begin{equation}\label{Kx}
f_i(x)=\frac{K_i x}{K_i+x},
\end{equation}
where $K_i$ represents the abundance of food at site $i$.
At the time $n+1$, the successful foragers will return to their 
respective foraging sites. Those that were unsuccessful will redistribute to different foraging directions.
\begin{equation}
u_{i,n+1}= f_i(u_{i,n})+(1-r_n) q_{i,n},
\end{equation}
where
\[
	r_n=\sum_{i=1}^m f_i(u_{i,n}), \quad q_{i,n} = \frac{\rho f_i(u_{i,n})}{r_n}+\frac{1-\rho}{m}.
\]
Here $\rho \in [0,1]$ quantifies the fraction of foragers recruited to the successful foraging sites. 
The rest of the unsuccessful foragers redistribute uniformly.
It is convenient to introduce vector notation. Let $\bm{u}_n=(u_{1,n},\cdots,u_{m,n})^{\rm T}$ and $\bm{f}(\bm{u}_{n})=(f_1(u_{1,n}),\cdots,f_m(u_{m,n}))^{\rm T}$.
Then, the above equation can be written as:
\begin{equation}\label{veceqn}
\bm{u}_{n+1}=\bm{f}(\bm{u}_n)+(1-\dual{\bm{1}}{\bm{f}(\bm{u}_n)})\paren{\frac{\rho}{\dual{\bm{1}}{\bm{f}(\bm{u}_n)}}\bm{f}(\bm{u}_n)+\frac{(1-\rho)}{m}\bm{1}}
\end{equation}
where $\bm{1}=(1,\cdots,1)^{\rm T}\in \mathbb{R}^m$ and $\langle \cdot, \cdot \rangle$ is the standard inner product. It is easily checked that:
\begin{equation}
\dual{\bm{1}}{\bm{u}_{n+1}}=\dual{\bm{1}}{\bm{u}_n}.
\end{equation}
We will assume, as in \eqref{udef}, that
\begin{equation}
\dual{\bm{1}}{\bm{u}_n}=1.
\end{equation}

\subsection{Model analysis when $\rho < 1$} 
Here we show that the infinite population model has a unique stable steady state when $0 \leq \rho < 1$, under some general conditions on $f_i$.

\subsubsection{Steady state solution}
First, we show that there is a unique steady state when $0 \leq \rho < 1$ under the following conditions for $f_i$: Let us assume that 
\begin{equation}\label{phicond}
f_i(0)=0, \; f_i(x)<x \text{ if } x>0,
\end{equation}
and the derivatives of $f_i(x)$ also satisfy
\begin{equation}\label{phiderivcond}
f_i'(x)>0,\; f_i''(x)<0.
\end{equation}
We finally assume that $f_i(x)$ saturates at high values of $x$
\begin{equation}\label{phiinfcond}
\lim_{x\to \infty} f_i(x)=K_i<\infty
\end{equation}
We note that the above two conditions imply the following:
\begin{equation}\label{phiderivinf}
\lim_{x\to \infty} f_i'(x)=0.
\end{equation}

The steady-state $\bm{u}\in \mathbb{R}^m$ satisfies:
\begin{equation}\label{uss}
	\bm{u}=\paren{(1-\rho)+\frac{\rho}{\mu}}\bm{f}(\bm{u})+\frac{1-\rho}{m}\paren{1-\mu}\bm{1}, \quad \mu=\dual{\bm{1}}{\bm{f}(\bm{u})}.
\end{equation}
Here  $\mu$ corresponds to the steady state of $r_n$. 
Note that $0<\mu<1$ by condition \eqref{phicond}. In component form, the first equation gives:
\begin{equation}\label{uieqn}
	u_i=\paren{1-\rho+\frac{\rho}{\mu}}f_i(u_i)+\frac{1-\rho}{m}(1-\mu) \equiv g_\mu(u_i).
\end{equation}
Then, for $0<\mu<1$, we have:
\begin{equation}
	g_\mu(0)=\frac{1-\rho}{m}(1-\mu)>0.
\end{equation}
Furthermore, given (\ref{phiderivcond}-\ref{phiderivinf}), $g_\mu(x)-x$ is an increasing function up to a certain value of $x=x_*$
($x_*$ can be $0$) after which it is a decreasing function that tends to $-\infty$ as $x\to \infty$. This implies that \eqref{uieqn} has a unique solution:
\begin{equation}
	u_i = h_i(\mu), \; h_i(\mu)=g_\mu(h_i(\mu)).
\end{equation}
At this point, let us note that:
\begin{equation}\label{sigma}
	\sigma_i=\at{\D{g_\mu}{x}}{x=h_i(\mu)}<1.
\end{equation}
This is because the derivative of $g_\mu(x)-x$ must be negative at $x=h_i(\mu)$.
We note that $h_i(\mu)$ is a decreasing function of $\mu$. Indeed, 
\begin{equation}
	\D{h_i}{\mu}=\PD{g_\mu}{\mu}+\at{\PD{g_\mu}{x}}{x=h_i(\mu)}\D{h_i}{\mu}.
\end{equation}
We see from \eqref{uieqn} that $\partial g_\mu/\partial \mu<0$ and by \eqref{sigma} we conclude that $\D{h_i}{\mu}<0$.
We also note that:
\begin{equation} \label{hlim}
	\lim_{\mu \to 1^-} h_i(\mu)=0.
\end{equation}
Now, let us substitute $u_i=h_i(\mu)$ into the second condition in \eqref{uss}.
\begin{equation}
\mu=\sum_{i=1}^m f_i(h_i(\mu)).
\end{equation}
Since $h_i(\mu)$ is positive for $0<\mu<1$ and is monotone decreasing with \eqref{hlim}, together with \eqref{phicond} and \eqref{phiderivcond}, we see that the above has a unique solution
in $0<\mu<1$.
This shows that the steady state is unique.

\subsubsection{Linear stability analysis}
Here we consider the stability of the steady state when $0 \leq \rho < 1$. If we linearize \eqref{veceqn} around the steady state, we have:
\begin{equation}
\begin{split}
\bm{w}_{n+1}&=\paren{A-\bm{\nu}\bm{\sigma}^{\rm T}}\bm{w}_n,\\
A&=\begin{pmatrix} 
\sigma_1 & 0 & \cdots  & 0\\
0 & \sigma_2 & \cdots & 0\\
\vdots & \vdots &  \ddots & \vdots \\
0 & 0 & \cdots & \sigma_m
\end{pmatrix}, \; 
\bm{\sigma}=\begin{pmatrix} \sigma_1\\ \sigma_2 \\ \vdots \\ \sigma_m
\end{pmatrix},\\
\bm{\nu}&=\paren{(1-\rho)+\frac{\rho}{\mu}}^{-1}\paren{\frac{\rho}{\mu}\bm{\lambda}+\frac{1-\rho}{m}\bm{1}}, \; \bm{\lambda}=\bm{\phi}(\bm{u})\mu^{-1}
\end{split}
\end{equation}
where $0<\sigma_i<1$ were defined in \eqref{sigma}. Note that all quantities are evaluated at the steady state satisfying \eqref{uss}.
We point out that:
\begin{equation}\label{nuprop}
\dual{\bm{\nu}}{\bm{1}}=1, \; \nu_i>0 \text{ where } \bm{\nu}=(\nu_1,\cdots,\nu_m)^{\rm T}.
\end{equation}
The equality above implies that $\bm{1}^{\rm T}$ is a left eigenvector of $A-\bm{\nu}\bm{\sigma}^{\rm T}$ with eigenvalue $0$.
To study the spectral properties of $A-\bm{\nu}\bm{\sigma}^{\rm T}$, consider the matrix:
\begin{equation}
Q=\begin{pmatrix} 
(\nu_1/\sigma_1)^{1/2} & 0 & \cdots  & 0\\
0 & (\nu_2/\sigma_2)^{1/2}& \cdots & 0\\
\vdots & \vdots &  \ddots & \vdots \\
0 & 0 & \cdots & (\nu_m/\sigma_m)^{1/2}
\end{pmatrix}
\end{equation}
We have:
\begin{equation*}
Q^{-1}(A-\bm{\nu}\sigma^{\rm T})Q=A-\bm{\alpha}\bm{\alpha}^{\rm T}, \; \bm{\alpha}=(\sqrt{\nu_1\sigma_1},\sqrt{\nu_2\sigma_2},\cdots,\sqrt{\nu_m\sigma_m})^{\rm T}.
\end{equation*}
So we may study the eigenvalues of $A-\bm{\alpha}\bm{\alpha}^{\rm T}$ instead of $A-\bm{\nu}\bm{\sigma}^{\rm T}$. First, note that 
$A-\bm{\alpha}\bm{\alpha}^{\rm T}$ is negative semi-definite. Indeed, for $\bm{x}=(x_1,\cdots,x_m)^{\rm T}\in \mathbb{R}^m$, we have:
\begin{equation*}
\begin{split}
\dual{\bm{x}}{(A-\bm{\alpha}\bm{\alpha}^{\rm T})\bm{x}}&=\dual{\bm{x}}{A\bm{x}}-\dual{\bm{\alpha}}{\bm{x}}^2
=\sum_{i=1}^m \sigma_ix_i^2-\paren{\sum_{i=1}^m \sqrt{\nu_i\sigma_i}x_i}^2\\
&\geq \sum_{i=1}^m \sigma_ix_i^2-\paren{\sum_{i=1}^m \nu_i}\paren{\sum_{i=1}^m \sigma_i x_i^2}=0,
\end{split}
\end{equation*}
where we used the Cauchy Schwarz inequality and \eqref{nuprop}. Furthermore, 
\begin{equation*}
\dual{\bm{x}}{(A-\bm{\alpha}\bm{\alpha}^{\rm T})\bm{x}}\leq \dual{\bm{x}}{A\bm{x}}\leq \sigma_{\rm max}\dual{\bm{x}}{\bm{x}}, \; \sigma_{\rm max}=\max_{1\leq i\leq m} \sigma_i.
\end{equation*}
This implies that the eigenvalues of $A$ are all non-negative and bounded above by $\sigma_{\rm max}<1$. This establishes local stability.
In fact, we can get more explicit information about the eigenvalues. Suppose $0<\sigma_m<\cdots < \sigma_2<\sigma_1<1$.
Let $z$ be an eigenvalue. Then,
\begin{equation}
\begin{split}
\det(A-\bm{\nu}\bm{\sigma}^{\rm T}-zI)&=\det(A-z I)\det (I-(A-z I)^{-1}\bm{\nu}\bm{\sigma}^{\rm T})\\
&=\det(A-zI)\paren{1-\dual{(A-zI)^{-1}\bm{\nu}}{\bm{\sigma}}}=0.
\end{split}
\end{equation} 
In the second equality, we used a one-rank update formula for the determinant. Thus, if we can find $m-1$ solutions to the equation:
\begin{equation*}
\paren{1-\dual{(A-zI)^{-1}\bm{\nu}}{\bm{\sigma}}}=1-\sum_{i=1}^m \frac{\nu_i\sigma_i}{z-\sigma_i}=0,
\end{equation*}
then we are done. Since $\nu_i\sigma_i>0$, the above has one solution in each interval $\sigma_{k}<z<\sigma_{k+1}, k=1,\cdots m-1$.

\subsection{Steady state solutions when $\rho = 1$} \label{sect2.2}
Unlike the above analysis, there exist multiple steady states when $\rho = 1$. We determine the steady states and show that one of the steady states can be obtained by taking limit when $\rho < 1$.

First, we view $\bm{u}(\rho)$ and hence $\mu(\rho)$ and other quantities as functions of $\rho$. At $\rho = 1$, \eqref{uss} can be written by $\mu \bm{u} - \bm{f}(\bm{u}) = 0$, which follows
\begin{equation} \label{sss_gencond1}
    u_i [\mu - \phi_i(u_i)] = 0
\end{equation}
If $\phi_i$ is invertible, the solution has to satisfy
\begin{equation} \label{sss_gencond2}
    \sum_{i \in \mathcal{I}} \phi_i^{-1}(\mu) = 1, \quad \mathcal{I} = \{ i : u_i \neq 0 \}.
\end{equation}
There exist $\mu \in (0,1)$ according to the intermediate value theorem.  
Imposing \eqref{Kx}, we have
\begin{equation} \label{ssrho1}
	u_i = \begin{cases}
	K_i/\sum_{i \in \mathcal{I}} K_i, & i \in \mathcal{I} \\
	0, & \text{otherwise}
	\end{cases}
\end{equation}
and
\begin{equation} \label{mu_rho1}
	\mu = \frac{\sum_{i \in \mathcal{I}} K_i}{1+ \sum_{i \in \mathcal{I}} K_i},
\end{equation}
for any non-empty index set $\mathcal{I} \subset \mathcal{I}_m \equiv \{1,\cdots,m\}$. Since $\mathcal{I}_m$ has multiple subsets, there are multiple steady states at $\rho = 1$. 
In other words, the steady state of infinite foragers depends on their initial allocation, so there are no foragers at some foraging sites if no foragers fly out of those sites.

Finally, we determine which of the steady state at $\rho = 1$ is continuous in $\rho$. That is, 
\begin{equation} \label{uss_rho1}
	\lim_{\rho \to 1} u_i(\rho) = \frac{K_i}{\sum_{i=1}^m K_i}, \quad \lim_{\rho \to 1} \mu(\rho) = \frac{\sum_{i=1}^m K_i}{1 + \sum_{i=1}^m K_i}.
\end{equation}
Suppose that the steady state's support is not $\mathcal{I}_m$ so that $u_i(1) = 0$ for some $i$. Let us note that
\begin{equation}
\begin{split}
	\at{\PD{g}{u_i}}{\rho = 1} &= \frac{1}{\mu(1)}\at{\D{f_i}{u_i}}{\rho=1}, \; \at{\D{f_i}{u_i}}{\rho = 1} = \phi_i^2(u_i(1)), \\
	\at{\PD{g}{\mu}}{\rho=1} &= -\frac{f_i(u_i(1))}{\mu^2(1)} = 0, \\
	\at{\PD{g}{\rho}}{\rho=1} &= (1-\mu(1))\left(\frac{f_i(u_i(1))}{\mu(1)} - \frac{1}{m}\right) = \frac{\mu(1)-1}{m}.
\end{split}
\end{equation}
Taking the derivative of \eqref{uieqn} to $\rho$ gives
\begin{equation}
	\at{\PD{u_i}{\rho}}{\rho=1} = \PD{g}{u_i}\at{\PD{u_i}{\rho}}{\rho=1} + \PD{g}{\mu} \sum_{j=1}^m \PD{\mu}{u_j}\at{\PD{u_j}{\rho}}{\rho = 1} + \at{\PD{g}{\rho}}{\rho = 1},
\end{equation}
which follows that 
\begin{equation}
	\at{\PD{u_i}{\rho}}{\rho=1} = \frac{\mu(1)}{m} > 0.
\end{equation}
This implies that $u_i(\rho) < 0$ in some neighborhood of $\rho = 1$, which is a contradiction. Therefore, we proved \eqref{uss_rho1}.

\subsection{Long-term reward rate}
Here we define the foraging reward rate and investigate its optimality, as a function of recruitment probability and distribution of site abundances. The foraging reward is determined by how many foragers are successful for each attempt. We define the foraging reward rate of the infinite population model by the stationary fraction of successful foragers
\begin{equation}
	\mu = \lim_{n \to \infty} r_n,
\end{equation}
which satisfies the steady state equation \eqref{uss}. We are interested in how much communication (or recruitment) maximizes the foraging reward rate.

More precisely, we find the maximum of $\mu$ with respect to $\rho$ using the Lagrangian multiplier method. Setting
\begin{equation} \label{wmaximization}
	\mathcal{L}(\bm{u},z) = \dual{\bm{1}}{\bm{f}(\bm{u})} - z ( \dual{\bm{1}}{\bm{u}} - 1),
\end{equation}
which maximum gives the optimal forager allocation without considering recruitment between foragers. This maximum should not be smaller than one with recruitment. Deriving the critical conditions for $\mathcal{L}$
\begin{equation}
	\PD{\mathcal{L}}{u_i} = \phi_i^2(u_i) - z = 0, \quad i = 1,\cdots,m.
\end{equation}
That is, the foraging process is optimized when foragers are allocated with the same foraging probability for all sites. A steady state solution at $\rho = 1$ satisfies (\ref{sss_gencond1}-\ref{sss_gencond2}) with $\mathcal{I} = \mathcal{I}_m$, which is the same as the above critical condition. Thus the reward rate of the infinite population model is maximized at $\rho = 1$.
In other words, the foraging process of the infinite population is most efficient under perfect recruitment.

Furthermore, in the special case when all foraging sites have the same foraging probability $\phi_i \equiv \phi$ for all $i$, then the reward rate does not depend on $\rho$ because the unique solution has to be uniform $u_i = 1/m$. One particular case is when the resource is uniformly distributed $K_i \equiv K$ for all $i$.
Numerical simulations in Fig. \ref{fig2} also show that the reward rate is maximized at $\rho = 1$. On the contrary, $\mu$ is insensitive to the choice of $\rho$ when $K_i$ is uniform in foraging sites, as seen in Fig. \ref{fig2}. In other words, the foraging efficiency does not matter when resources are evenly distributed.
\begin{figure}[t!]
\centering
\includegraphics[width=.8\columnwidth]{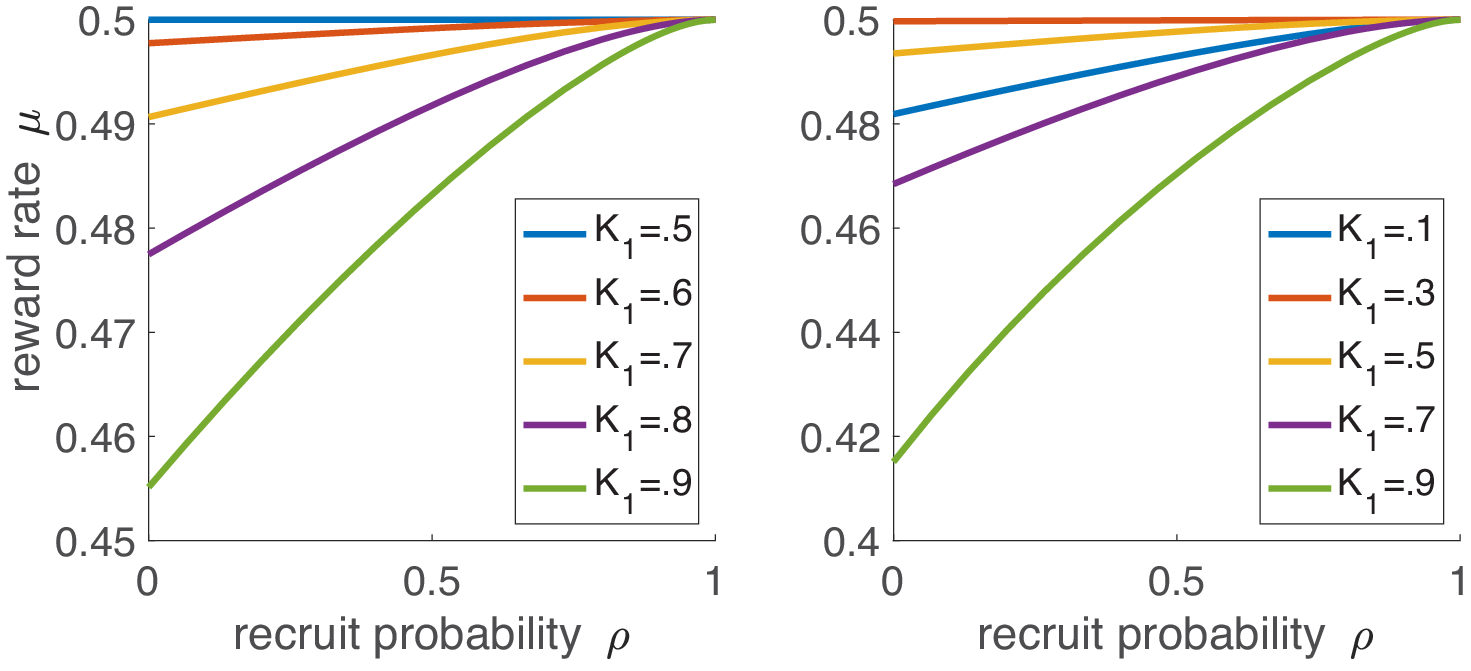}
\caption{\label{fig2} Reward rate of infinite population model as a function of $\rho$ with various parameters. (a) Bidirectional case $m = 2$ where $K_1 + K_2 = 1$. (b) Tri-directional case $m = 3$ where $K_1+K_2 +K_3 = 1$ and $K_2 = K_3$.} 
\end{figure}

An additional observation is that 
\begin{equation}
	\at{\PD{\mu}{\rho}}{\rho = 1} = 0,
\end{equation}
for any $m$. In other words, regardless of the number of total foraging sites, the foraging efficiency is saturated with perfect communication between foragers, which also can be seen in Fig. \ref{fig2}.
This can be shown by the derivative of the second equation in \eqref{uss} with respect to $\rho$.
Since $\phi_i(u_i(1)) = \mu(u_i(1))$, we have
\begin{equation}
	\at{\D{\mu}{\rho}}{\rho=1}=\sum_{i=1}^m \at{\D{f_i}{u_i}\D{u_i}{\rho}}{\rho=1}
=\mu^2(1)\sum_{i=1}^m \at{\D{u_i}{\rho}}{\rho=1}=0
\end{equation}
where in the last equality, we used the fact that the $\sum_{i=1}^m u_i=1$ regardless of $\rho$.

\section{Finite population model} \label{sect3} Now we introduce a finite population model analogous to the infinite-population model analyzed above.
In contrast to the infinite population model in Sect. \ref{sect2}, we have to track intrinsic fluctuations in the case of a finite population.
We cannot impose mass-action principles for quantifying successful foragers and redistributing unsuccessful foragers.
In this section, we introduce a discrete-time Markov chain modeling the foraging process with a finite population. We then define and investigate the reward rate corresponding to one for the infinite population model.

We consider a stochastic version of the foraging model with intrinsic noise due to the discreteness of the dynamics in finite population $\xi$, as illustrated in Fig. \ref{fig3}. 
\begin{figure}[t!]
\centering
\includegraphics[width=.5\columnwidth]{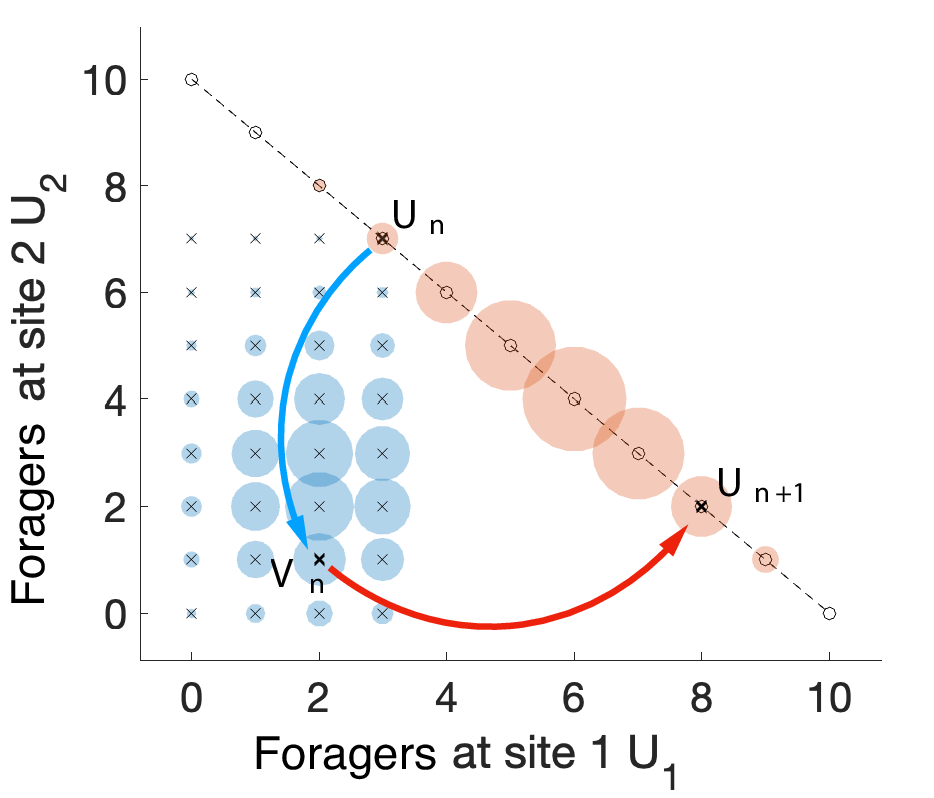}
\caption{\label{fig3} Illustration of discrete-time Markov chain of foraging process by finite population when $m = 2$. The size of a circle represents the probability of transitioning to the corresponding state.} 
\end{figure}
Let $U_{i,n}$ be the number of foragers that fly to site $i = 1,\cdots,m$ at time $n = 0,1,\cdots$. We assume that $\xi$ foragers fly out every time, so we have
\begin{equation} \label{Udef}
	\sum_{i = 1}^m U_{i,n} = \xi, \; U_{i,n} \geq 0.
\end{equation}
We impose an initial state $U_{i,0}$ satisfying \eqref{Udef}.
Of those foragers, those that successfully forage satisfies
\begin{equation} \label{veqn}
	V_{i,n} \sim B(U_{i,n},\phi_i(U_{i,n}/\xi)),
\end{equation}
where $B(x,p)$ is the binomial distribution with $x$ trials and probability $p$. Here we assume that foragers try to forage independently with the identical probability defined in \eqref{success_prob}. Then the total number of successful foragers at time $n$ takes the form
\begin{equation}
	R_n = \sum_{i=1}^m V_{i,n}.
\end{equation}
In contrast to the deterministic model, all the foragers may fail to forage, and thus $R_n = 0$. In this case, we assume that the foragers are reset to the uniform search at the following time
\begin{equation} \label{all_fail}
	q_{i,n} = \frac{1}{m}, \; \text{if } R_n = 0.
\end{equation}
Otherwise, the successful foragers will fly out to the same sites at the next time step, and the unsuccessful foragers will redistribute with probability
\begin{equation} \label{redi_prob}
	q_{i,n} = \rho \frac{U_{i,n}}{R_n} + (1-\rho) \frac{1}{m}, \; \text{if } R_n \neq 0.
\end{equation}
It is convenient to introduce vector notation. Let $\bm{U}_n = (U_{1,n},\cdots,U_{m,n})^T$, $\bm{V}_n = (V_{1,n},\cdots,V_{m,n})^T$, and $\bm{q}_n = (q_{1,n},\cdots,q_{m,n})^T$. Then, the foragers that fly out at time $n+1$ satisfies
\begin{equation} \label{stoc_veceqn}
	\bm{U}_{n+1} - \bm{V}_n \sim M_m(\xi - R_n,\bm{q}),
\end{equation}
where $M_m(x,\bm{q})$ is the multinomial distribution with $x$ trials to $m$ sites with probability $\bm{q}$. One can show that
\begin{equation}
	\dual{\bm{1}}{\bm{U}_{n+1}}=\dual{\bm{1}}{\bm{U}_n} = \xi,
\end{equation}
according to the assumption \eqref{Udef}. 

Here note that the stochastic recurrence relation \eqref{stoc_veceqn} converges (almost surely) to the deterministic recurrence relation \eqref{veceqn} as $\xi \to \infty$ by setting $u_{i,n} = U_{i,n}/\xi$, $v_{i,n} = V_{i,n}/\xi$, and $r_n = R_n/\xi$, according to the law of large numbers. 
However, this is not the only formulation that converges to the same deterministic limit. For example, instead of allocating the successful foragers deterministically, the entire set of foragers can be redistributed by the following:
\begin{equation} \label{model2}
	\bm{U}_{n+1} \sim M_m(\xi,\tilde{\bm{q}}), \quad \tilde{\bm{q}} = \frac{R_n}{\xi} \frac{U_{i,n}}{R_n} + \left(1 - \frac{R_n}{\xi} \right) \bm{q}.
\end{equation}
Since this formulation is based on an unrealistic assumption (the successful foragers have no memory of successful sites, but the entire set of foragers do), we choose the formulation in the above paragraph instead.

Our stochastic model also can be formulated by the Master equation. Let $\bm{x} = (x_1,\cdots,x_m)^T$ and let $p_n(\bm{x})$ be the probability that $x_i$ foragers are flying out to site $i$ at time $n$. Let $A(\bm{x}|\bm{x}')$ be the transition probability from state $\bm{U}_{n} = \bm{x}'$ to $\bm{U}_{n+1} = \bm{x}$. 
Then the master equation takes the form
\begin{equation} \label{Ap}
	p_{n+1}(\bm{x})= \sum_{\bm{x}' \in \mathcal{S}} A(\bm{x}|\bm{x}') p_n(\bm{x}'),
\end{equation}
where the state space preserves the total population $\mathcal{S} = \{ \bm{x} : \dual{\bm{1}}{\bm{x}} = \xi, x_i \geq 0\}$. We now determine the explicit form of  $A(\bm{x}|\bm{x}')$. 
We denote $\bm{y} = (y_1,\cdots,y_m)^T$. The probability that $y_i$ out of $x_i$ foragers forage successfully follows the binomial distribution
\begin{equation}
	\mathbb{P}[V_{i,n} = y_i | U_{i,n} = x_i] = b_2(y_i;x_i,\phi_i(x_i/\xi)),
\end{equation}
where $b_2(y;x,p) = {x \choose y} p^y (1-p)^y$. Assuming that the foraging process is independent of the sites, the probability of overall successful foragers takes the form
\begin{equation}
	a_1(\bm{y}|\bm{x}) \equiv \mathbb{P}[ \bm{V}_n = \bm{y} | \bm{U}_n = \bm{x}] = \prod_{i=1}^m b_2(y_i;x_i,\phi_i(x_i/\xi)).
\end{equation}
Since the unsuccessful foragers are redistributed by the multinomial distribution, the conditional transition probability with given successful forager allocation $\bm{y}$ satisfies
\begin{equation}
	a_2(\bm{x}|\bm{y}) \equiv \mathbb{P}[\bm{U}_{n+1} = \bm{x} | \bm{V}_n = \bm{y}] = b_m(\bm{x} - \bm{y};\xi - r, \bm{q}), 
\end{equation}
where $r = \dual{\bm{1}}{\bm{y}}$ and $\bm{q}$ is the redistribution probability defined in \eqref{redi_prob}. Here the explicit form of the multinomial distribution is
\[
	b_m(\bm{z};x,\bm{q}) = x! \prod_{i=1}^m \frac{q_i^{z_i}}{z_i!},
\]
if $0\leq z_i \leq x_i$ and $\dual{\bm{1}}{\bm{z}} = x$, otherwise zero.
We finally determine the explicit form of the transition probability
\begin{equation}
	A(\bm{x}|\bm{x}') = \sum_{0 \leq \bm{y} \leq \bm{x}'} a_2(\bm{x}|\bm{y}) a_1(\bm{y}|\bm{x}'),
\end{equation}
where $\sum_{0\leq \bm{y} \leq \bm{x}} = \sum_{0 \leq y_1 \leq x_1} \cdots \sum_{0 \leq y_m \leq x_m}$.

\subsection{Stationary distribution}
Here we derive an equation for the stationary distribution and prove that there exists a unique stationary distribution for any $\rho \in [0,1]$ by the Perron-Frobenius theorem (PFT). The stationary distribution of the stochastic model satisfies the following linear system
\begin{equation} \label{sdeqn}
	\pi(\bm{x}) = \sum_{\bm{x}' \in \mathcal{S}} A(\bm{x}|\bm{x}') \pi(\bm{x}').
\end{equation}
Let $\bm{\pi} = (\pi(\bm{x}))_{\bm{x} \in \mathcal{S}}^T$ and we introduce the matrix-vector notation for \eqref{sdeqn}
\begin{equation} \label{sdeqn_mat}
	\bm{\pi} = \bm{A} \bm{\pi}.
\end{equation}
We want to show that the linear system has a unique solution (up to $\dual{\bm{1}}{\bm{\pi}} = 1$). We first show that $A(\bm{x}|\bm{x}') >0$ for any $\bm{x},\bm{x}' \in \mathcal{S}$. Utilizing the total failure event \eqref{all_fail}, we have
\begin{align}
	A(\bm{x}|\bm{x}') &\geq a_2(\bm{x}|\bm{0}) a_1(\bm{0} | \bm{x}') \nonumber \\
	&= b_m(\bm{x};\xi,\bm{1}/m) \prod_{i=1}^m [1-\phi_i(x_i'/\xi)]^{x_i'} >0,
\end{align}
for any $\rho \in [0,1]$. According to the PFT, a positive matrix $\bm{A}$ has a simple eigenvalue (or a simple root of the characteristic polynomial of $\bm{A}$ and thus its eigenspace is one-dimensional), and the eigenvalue is the same as the spectral radius of $\bm{A}$. Since the spectral radius of any stochastic matrix (such as $\bm{A}$) is one, the solution space of \eqref{sdeqn_mat} is one-dimensional. This implies that the dimension of the solution space of the linear system \eqref{sdeqn_mat} is one-dimensional. Furthermore, the PF eigenvector is also positive, so the stochastic model has a unique stationary distribution for any $\rho \in [0,1]$.

What is the population limit of the stationary distribution? If $0 \leq \rho < 1$, the deterministic model \eqref{veceqn} has a unique fixed point $\bm{u}$, which is linearly stable. We expect that the population limit of the stationary distribution converges to the fixed point and this is confirmed by numerical simulations shown in Fig. \ref{fig4} and \ref{fig7}. That is, we expect
\begin{equation} \label{poplim_pdf}
    \lim_{\xi \to \infty} \sum_{\bm{x}/\xi \in S} \pi(\bm{x}) = \mathcal{I}_{\bm{u}}(S) \equiv \begin{cases}
        1, & \bm{u} \in S \\
        0, & \bm{u} \notin S
    \end{cases},
\end{equation}
where $\mathcal{I}$ is an indicator function. 
The main obstacle in proving this statement is showing the global stability of the deterministic fixed point $\bm{u}$.
Note that when $\rho = 1$ we have multiple steady states for the deterministic process, and so the limiting stationary distribution is not as simple as when $0 \leq \rho < 1$.

\subsection{Expectation of long-term reward rate}
Similar to the infinite population model, we define the reward rate for the finite population model. Since the number of successful foragers $R_n$ is now a random variable and depends on the total population $\xi$, we define the reward rate by the asymptotic expected fraction of successful foragers
\begin{equation}
	\mu_\xi = \lim_{n \to \infty} \frac{\mathbb{E}[R_n]}{\xi}.
\end{equation}
The reward rate can be determined by the stationary distribution
\begin{align}
	\mu_\xi &= \lim_{n \to \infty} \frac{1}{\xi}  \sum_{\bm{x} \in \mathcal{S}} \sum_{0 \leq \bm{y} \leq \bm{x}}\dual{\bm{1}}{\bm{y}} \mathbb{P}[ \bm{V}_n =\bm{y} | \bm{U}_n = \bm{x}] \mathbb{P}[\bm{U}_n = \bm{x}] \nonumber \\
	&= \sum_{\bm{x} \in \mathcal{S}} \sum_{i=1}^m f(x_i/\xi) \pi(\bm{x}).
\end{align}

One interesting feature of the stochastic model is that $\mu_\xi(\rho)$ can be non-monotonic and thus have an intermediate maximum (denoted by $\rho^\star$). In contrast, the infinite population model (or deterministic limit) has a trivial maximum at $\rho = 1$. 
Numerical simulation in Fig. \ref{fig4}(a) shows that $\mu_\xi(\rho)$ can be non-monotonic with a sufficiently large $\xi$ if $\mu(\rho)$ is an increasing function (This statement will be proven for any $m$ later in Sect. \ref{sect4}).
In other words, for any finite population, the foraging process somehow becomes inefficient with high (close to perfect) recruitment, which contrasts sharply with the infinite-population case.

The inefficiency of a finite population with high recruitment rate is also observed when the resource is uniformly distributed and $\mu(\rho)$ is constant, as seen in Fig. \ref{fig4}(b).
We investigate how this inefficiency arises in a finite population in Sect. \ref{sect4}.
\begin{figure}[t!]
\centering
\includegraphics[width=.8\columnwidth]{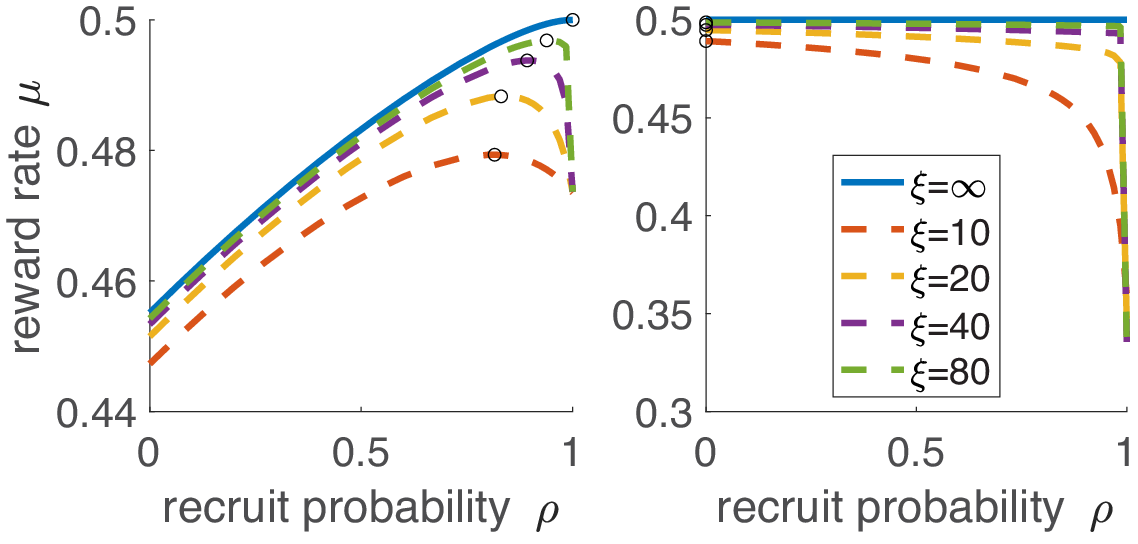}
\caption{\label{fig4} Expectation of reward rate of finite population model as a function of $\rho$ with various population size $\xi$. The reward rate is maximized at an intermediate optimal recruitment probability $\rho^\star$ (\textit{black circles}). (a) Biased resource abundance. $K_1 = 0.9$. (b) Unbiased resource abundance. $K_1 = 0.5$. Parameters as follows: $m = 2$ and $K_1 + K_2 = 1$.} 
\end{figure}
\begin{figure}[t!]
\centering
\includegraphics[width=.5\columnwidth]{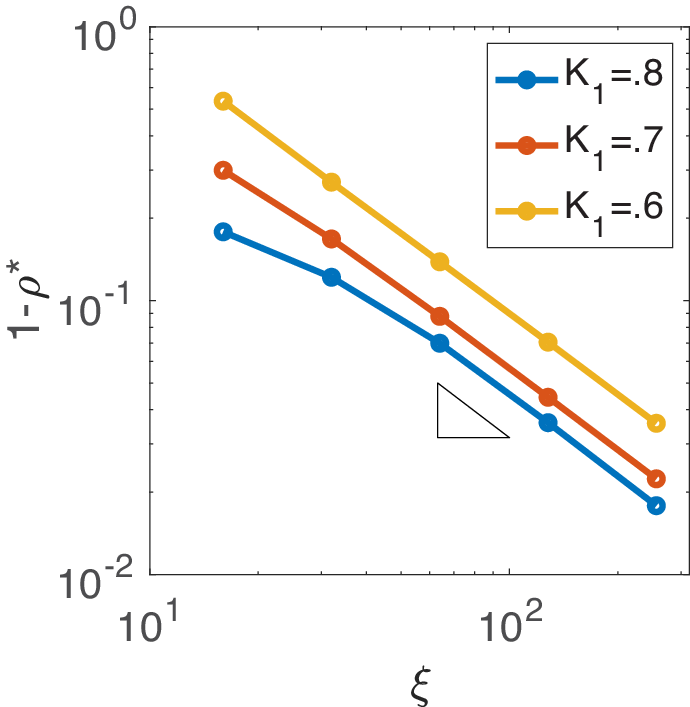}
\caption{\label{fig5} Convergence of $\rho^\star$ in $\xi$ with various $K_i$ when $m=2$. The slope of the triangle is $-1$.} 
\end{figure}
How does $\rho^*$ converge as $\xi \to \infty$? The optimal recruitment probability depends on the population size and it converges to the trivial maximum ($\rho^* = 1$) as $\xi \to \infty$ except for the uniform resource distribution ($\rho^* = 0$), as seen in Fig. \ref{fig4}. Furthermore, Fig. \ref{fig5} shows that
\begin{equation} \label{asympt}
    1 - \rho^* = \mathcal{O}(\xi^{-1}),
\end{equation}
when $m = 2$. 

\subsection{Relaxation time}
The stationary reward rate allows us to understand the optimal foraging strategy for a long (or infinite) period. But what is behavior of a group of foragers over a finite period of time? According to the exploration-exploitation trade-off, the foraging efficiency also depends on the convergence rate of the foraging system \eqref{Ap}. 
Suppose that $\bm{A}$ has eigenvalues $1 > \lambda_2 \geq …. \geq \lambda_{|\bm{A}|}$ where $|\bm{A}|$ is the size of square matrix $\bm{A}$ and the corresponding eigenvectors $\bm{v}_1, \bm{v}_2, \cdots, \bm{v}_{|\bm{A}|}$. We know that the largest eigenvalue is $1$ because \eqref{Ap} always has a unique stationary distribution, as also seen in Fig. \ref{fig_supp1}. Then the solution for the master equation can be written by
\begin{equation}
    \bm{p}_n = \bm{\pi} + \sum_{i=2}^{|\bm{A}|} c_i \lambda_i^n \bm{v}_i,
\end{equation}
which follows that the convergence rate is determined by the second largest eigenvalues $\lambda_2$. Fig. \ref{fig_supp1} shows that $\lambda_2$ increases in $\rho$ because a low recruitment probability lead to ``diffuse'' foragers over the foraging sites immediately. Therefore, if the foraging time is sufficiently short, the net reward over finite foraging time can be maximized at $\rho < \rho^*$. For the same reason, in the case of the infinite-population model, the optimal recruitment probability for a short time is smaller than $\rho^*$.

\begin{figure}[t!]
\centering
\includegraphics[width=.8\columnwidth]{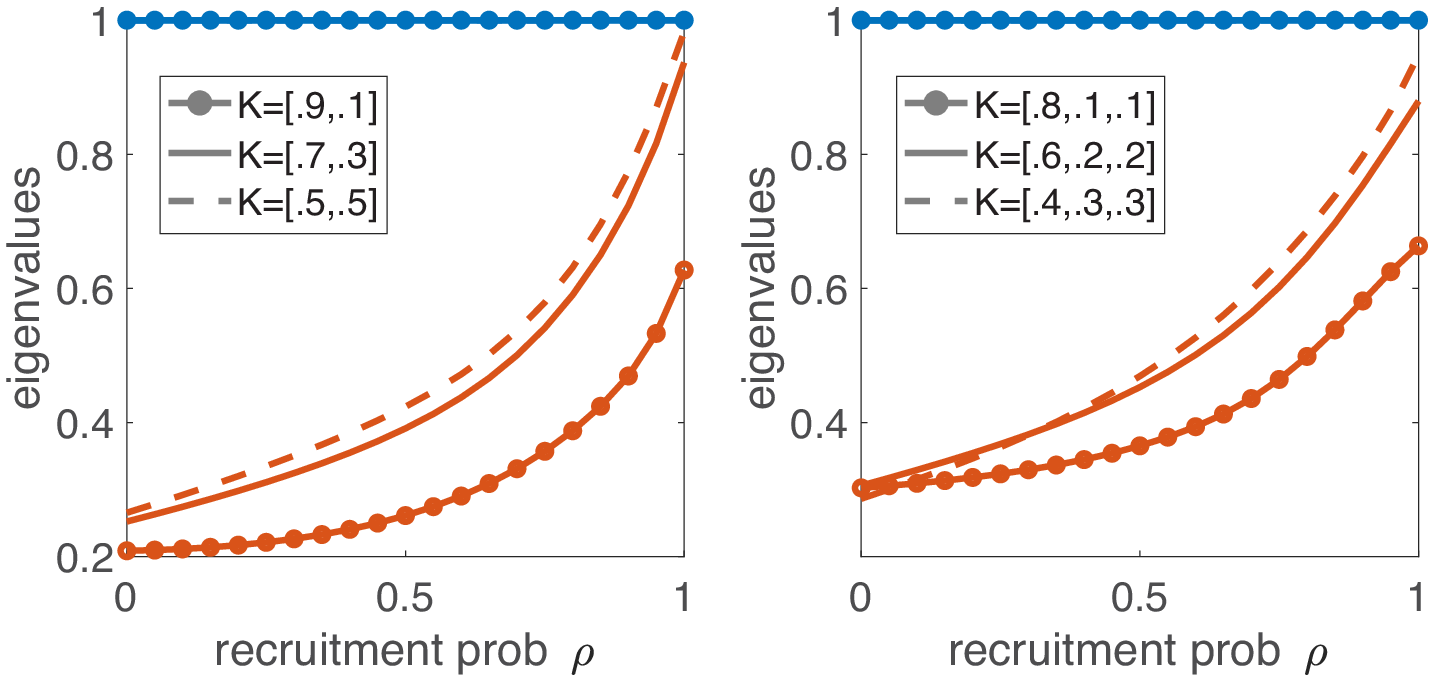}
\caption{\label{fig_supp1} The first (\textit{blue}) and the second largest eigenvalues (\textit{red}) of the transition matrix $\bm{A}$ in the master equation \eqref{Ap}. (a) Two foraging sites with various $K$. (b) Corresponding plot for three foraging sites. $\xi = 10$.} 
\end{figure}

\section{Quasi-steady states in finite population model} \label{sect4}
Why is the foraging process inefficient at a high recruitment probability and a large population?
The answer is tightly related to how the stationary distribution $\pi$ at $\rho = 1$ differs from that in $\rho < 1$.
Both distributions concentrate near the stable fixed points for the deterministic model, but there are multiple when $\rho = 1$, as shown in \eqref{mu_rho1}. We call these states as quasi-steady states (QSS) because $\pi(\bm{x})$ is proportional to the dwell time for given state $\bm{x}$.
Fig. \ref{fig6} and \ref{fig7} shows that the stationary distribution concentrates at the QSSs. 
However, the stationary distribution concentrates only at the most resource-abundant QSS as $\xi \to \infty$, whereas the others decays. In other words, all foragers are allocated to the most abundant resource site for the most of the time, which leads to overpopulation with low foraging probability. This inefficiency can be improved by decreasing recruitment probability $\rho$ because it weakens the foragers' collective behavior.

\begin{figure}[t!]
\centering
\includegraphics[width=.8\columnwidth]{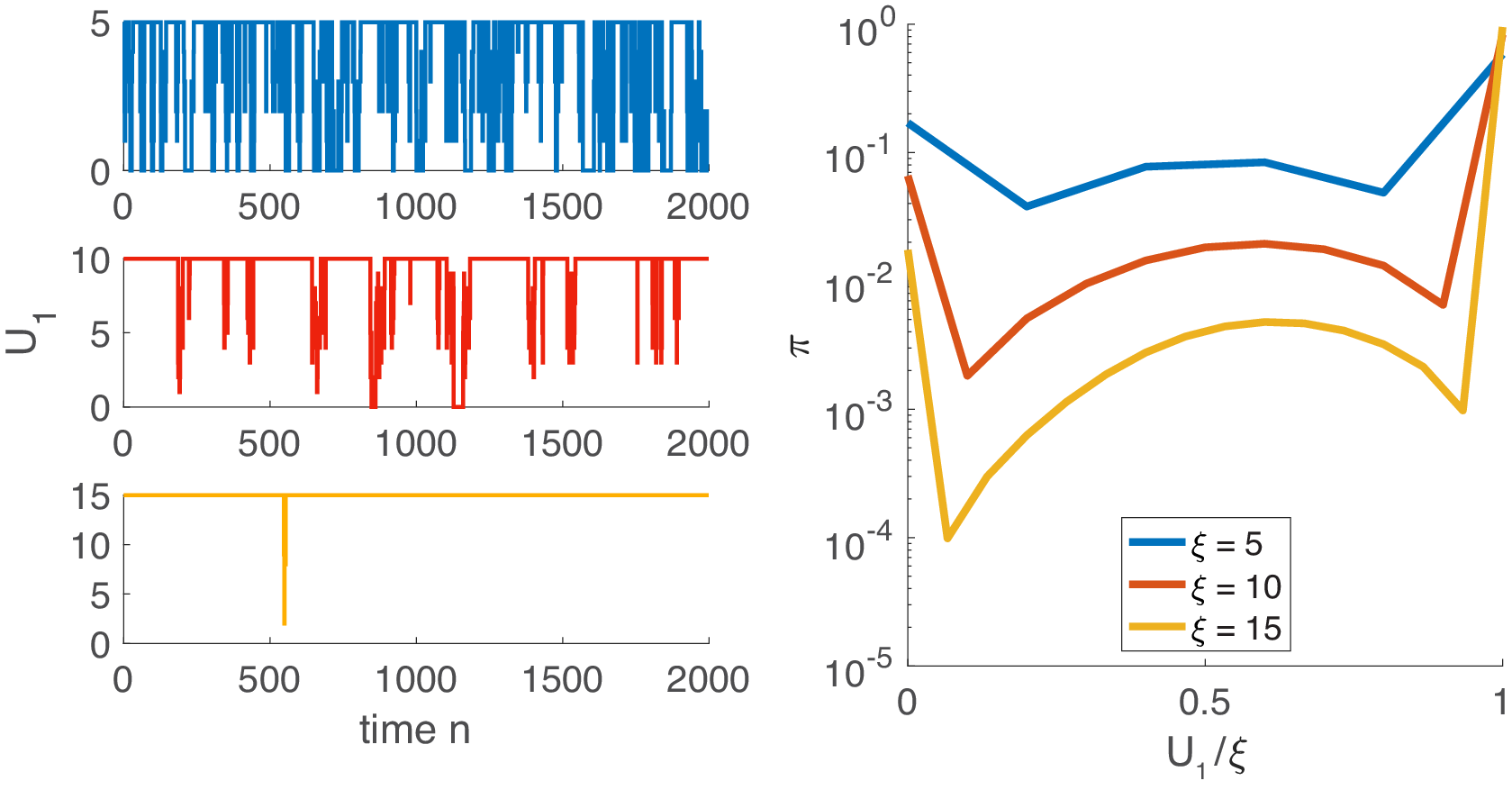}
\caption{\label{fig6} Quasi-steady state in finite population model when $\rho = 1$. (a) Sample trajectories with various $\xi$ when $m = 2$. (b) Corresponding plots of stationary distribution by solving \eqref{sdeqn}. Parameters are as follows: $K_1 = 0.9$ and $K_2 = 0.1$.} 
\end{figure}
\begin{figure}[t!]
\centering
\includegraphics[width=.8\columnwidth]{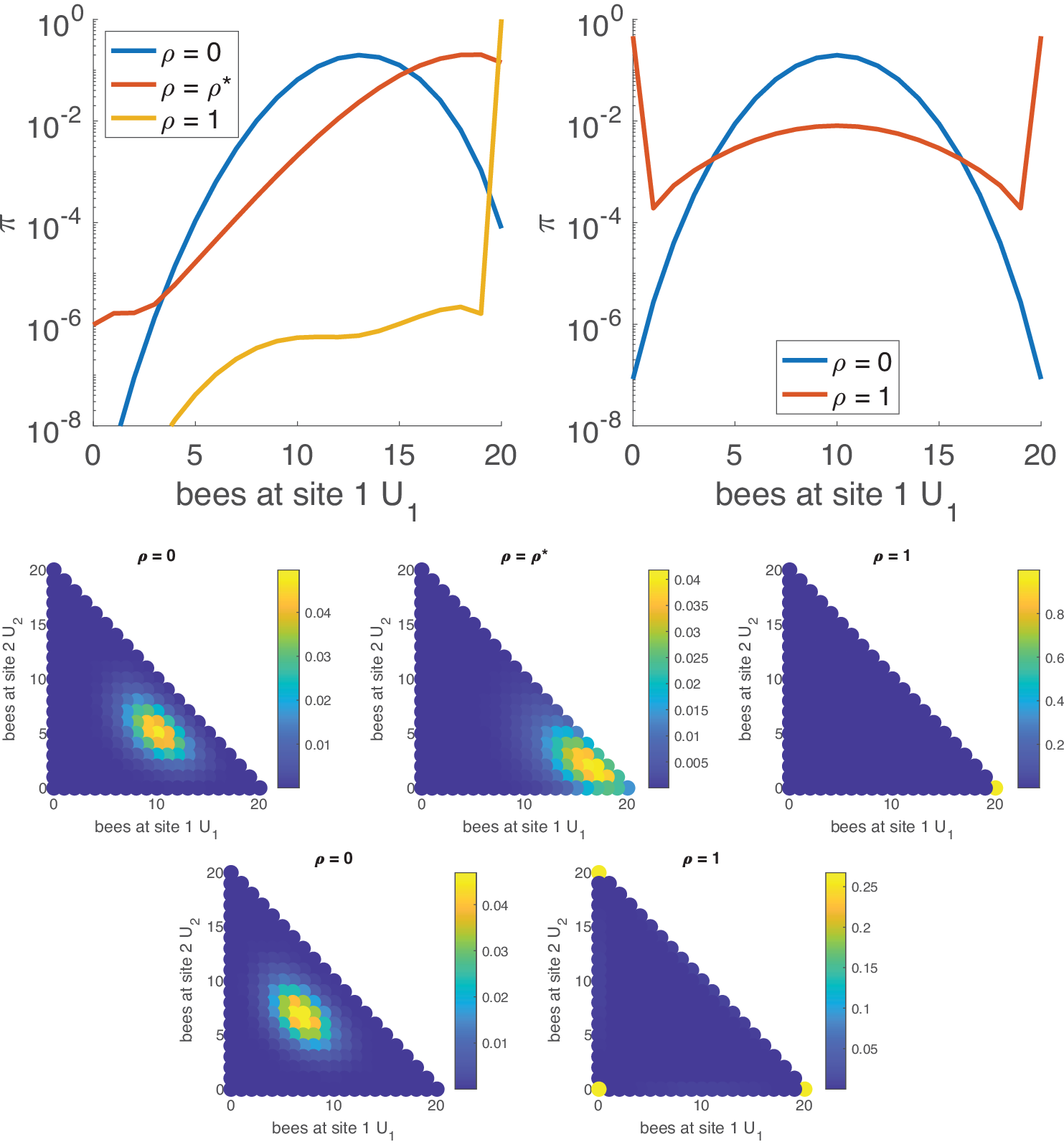}
\caption{\label{fig7} Stationary distribution of finite population model with various parameters. (a) Biased (\textit{left}, $K_1 = 0.9$) and unbiased (\textit{right}) resource distribution when $m = 2$. (b) Biased (\textit{top}, $K_1 = 0.8, K_2 = K_3 = 0.1$) and unbiased (\textit{bottom}) resource distribution when $m = 3$. Other parameter as follows: $\xi = 20$.} 
\end{figure}

In this chapter, we investigate this inefficiency by analyzing the stationary distribution at $\rho = 1$. First, we approximate the stationary distribution by the mean escape time from the QSSs and see why the most resource-abundant site is preferred by high recruitment foragers. Next, we prove that the distribution concentrates on the ``edges'' of the state space, which includes most of the QSSs. Therefore, the reward rate $\mu_\xi(\rho)$ can have a boundary layer at $\rho = 1$, which explains the inefficiency at high recruitment probability and an intermediate optimal recruitment $\rho^* < 1$.

\subsection{Mean escape time and stationary distribution} \label{sect4.1}
To understand intuitively why stationary distribution concentrates at the most resource-abundant foraging site when $\rho = 1$, it is approximated by in terms of the escape time from the QSSs.
The escape event happens only when all foragers fail to forage, and thus relocated uniformly.
Let $\bm{x} = \xi \bm{u}$ be a QSS. Then the escape probability from the QSS can be written by
\begin{equation}
    P_\text{esc}(\bm{x}) = \left[\prod_{i=1}^m \left(1- \phi_i(u_i)\right)^{u_i}\right]^\xi \equiv [\psi(\bm{u})]^\xi,
\end{equation}
which follows the mean escape time
\begin{equation}
    \tau_\text{esc}(\bm{x}) = P_\text{esc}^{-1}(\bm{x}).
\end{equation}
After escape, it relocates to a QSS with different probability quickly at large $\xi$. 
Assuming that the relocation probabilities are the same and the dwell times for non-QSS states are negligible, the stationary distribution can be approximated by
\begin{equation}
    \pi(\bm{x}) \approx \frac{\tau_\text{esc}(\bm{x})}{\sum_{\bm{x}'} \tau_\text{esc}(\bm{x}')} = \frac{[\psi(\bm{u})]^{-\xi}}{\sum_{\bm{u}'} [\psi(\bm{u}')]^{-\xi}} \equiv \pi_\text{approx}(\bm{u},\xi).
\end{equation}
Suppose that $K_1 > K_2 > \cdots > K_m \geq 0$. Since $\psi(\bm{u})$ is maximized at $\bm{u}_0 = (1,0,\cdots,0)^T$, the limit of the stationary distribution approximation goes to
\begin{equation}
    \lim_{\xi \to \infty} \pi_\text{approx}(\bm{u},\xi) = \mathcal{I}_{\bm{u}_0}(\bm{u}),
\end{equation}
where $\mathcal{I}$ is the indicator function.
If $K_i$ are not strictly ordered and $K_1 = K_2 =\cdots = K_{m'}$, then we have
\begin{equation}
    \lim_{\xi \to \infty} \pi_\text{approx}(\bm{u},\xi) = \frac{1}{m'}\sum_{i=1}^{m'}\mathcal{I}_{\bm{u}_i}(\bm{u}),
\end{equation}
where the $i$th entry of vector $\bm{u}_i \in [0,1]^m$ is $1$ and the other entries are $0$. 
Overall, the stationary distribution is high at maximum $K_i$ because it is harder to escape from that foraging site than the others.

\subsection{Stationary distribution convergence on edge states}

Instead of approximation, we prove that the stationary distribution concentrates on set $\mathcal{E}$ that includes most of the QSSs as $\xi \to \infty$. 
We introduce the following sets
\[
    \mathcal{B} = \{ \bm{x} \in \mathcal{S} | x_i > 0 \text{ for all } 1 \leq i \leq m\}, \quad \mathcal{E} = \mathcal{S} \backslash \mathcal{B},
\]
which separate the state space $\mathcal{S}$ into its bulk $\mathcal{B}$ and edges $\mathcal{E}$.
Note that all QSSs are included in $\mathcal{E}$ except one in \eqref{uss_rho1}. In this section, we want to prove 
\begin{equation} \label{conv_weak}
    \lim_{\xi \to \infty} \sum_{\bm{x} \in \mathcal{E}} \pi(\bm{x}) = 1.
\end{equation}

We begin our proof by decomposing \eqref{sdeqn_mat} into the following form: 
\begin{equation}\label{ABE}
\begin{pmatrix}
A_{\mc{B}\mc{B}} & A_{\mc{B}\mc{E}}\\
A_{\mc{E}\mc{B}} & A_{\mc{E}\mc{E}}
\end{pmatrix}
\begin{pmatrix}
\bm{\pi}_{\mc{B}}\\
\bm{\pi}_{\mc{E}}
\end{pmatrix}
=\begin{pmatrix}
\bm{\pi}_{\mc{B}}\\
\bm{\pi}_{\mc{E}}
\end{pmatrix}, \quad \bm{\pi}=\begin{pmatrix}
\bm{\pi}_{\mc{B}}\\
\bm{\pi}_{\mc{E}}
\end{pmatrix}
, \quad A=\begin{pmatrix}
A_{\mc{B}\mc{B}} & A_{\mc{B}\mc{E}}\\
A_{\mc{E}\mc{B}} & A_{\mc{E}\mc{E}}
\end{pmatrix}.
\end{equation}
Let us also introduce the notation $\norm{\cdot}_1$ to denote the vector and matrix $1$-norms. That is to say, for a vector $\bm{v}\in \mathbb{R}^N$ and a $M\times N$ matrix $C$,
\begin{equation}
\norm{v}_1=\sum_{k=1}^N \abs{v_i}, \; \norm{C}_1=\max_{1\leq j\leq N} \sum_{i=1}^M \abs{c_{ij}},
\end{equation}
where $v_i$ and $c_{ij}$ are the elements of the vector and matrix respectively. Our goal is equivalent to show that:
\begin{equation}
\lim_{\xi \to \infty }\norm{\bm{\pi}_{\mc{B}}}_1=0.
\end{equation}
From \eqref{ABE}, we have the equation:
\begin{equation}
A_{\mc{B}\mc{B}}\bm{\pi}_{\mc{B}}+A_{\mc{B}\mc{E}}\bm{\pi}_{\mc{E}}=\bm{\pi}_{\mc{B}}.
\end{equation}
From this, we see that:
\begin{equation}
\bm{\pi}_{\mc{B}}=(I-A_{\mc{B}\mc{B}})^{-1}A_{\mc{B}\mc{E}}\bm{\pi}_{\mc{E}}.
\end{equation}
Noting that $\norm{\bm{\pi}_{\mc{E}}}_1\leq 1$, we have:
\begin{equation}\label{piBest}
\norm{\bm{\pi}_{\mc{B}}}_1\leq \norm{(I-A_{\mc{B}\mc{B}})^{-1}}_1\norm{A_{\mc{B}\mc{E}}}_1\leq \frac{1}{1-\norm{A_{\mc{B}\mc{B}}}_1}\norm{A_{\mc{B}\mc{E}}}_1.
\end{equation}
We will now estimate $\norm{A_{\mc{B}\mc{E}}}_1$ and $\norm{A_{\mc{B}\mc{B}}}_1$. 
We first consider $A_{\mc{B}\mc{E}}$. Take a column vector of $A_{\mc{B}\mc{E}}$ that corresponds to state $\bm{x}\in \mc{E}$, and call it $\bm{p}_{\bm{x}}$.
The $1$-norm of $\bm{p}_{\bm{x}}$ is the probability that we transition from $\bm{U}_n=\bm{x}=(x_1,\cdots,x_m)^{\rm T}\in \mc{E}$ 
to one of the bulk states  so that $\bm{U}_{n+1}\in \mc{B}$.
Since at least one of the $x_i=0$, this can only happen if $V_{i,n}=0$ for all $i$ (no foragers are successful).
Thus:
\begin{equation}
\begin{split}
\norm{\bm{p}_{\bm{x}}}_1&\leq\mathbb{P}[V_{i,n}=0 \text{ for all } 1\leq i\leq m | \bm{U}_n=\bm{x}] \\
&=\prod_{i=1}^m (1-\phi_i(x_i/\xi))^{x_i}=\prod_{i=1}^m \paren{\frac{x_i/\xi}{(x_i/\xi)+K_i}}^{x_i}
\end{split}
\end{equation}
where we adopt the convention $0^0=1$. Note that:
\begin{equation}\label{alphastar}
\frac{x_i/\xi}{(x_i/\xi)+K_i}\leq \frac{1}{1+K_i}\leq \frac{1}{1+K_{\rm min}}=\alpha_*, \quad K_{\rm min}=\min_{1\leq i\leq m} K_i.
\end{equation}
Thus,
\begin{equation}
\norm{\bm{p}_{\bm{x}}}_1\leq \prod_{i=1}^m \alpha_*^{x_i}=\alpha_*^{\xi}.
\end{equation}
Therefore,
\begin{equation}\label{ABEest}
\norm{A_{\mc{B}\mc{E}}}_1=\max_{\bm{x}\in \mc{E}}\norm{\bm{p}_{\bm{x}}}_1\leq \alpha_*^\xi
\end{equation}
Now, we turn to the estimation of $\norm{A_{\mc{B}\mc{B}}}_1$. 
Let $\bm{q}_{\bm{x}}$ be the column vector of $A_{\mc{B}\mc{B}}$ corresponding to state $\bm{x}\in \mc{B}$.
The $1$-norm of $\bm{q}_{\bm{x}}$ is the probability that, starting at a state $\bm{x}\in \mc{B}$, you are back in one of the states in $\mc{B}$.
This can happen in two ways. The first way in which this can happen is that $V_{i,n}\geq 1$ for all $i$ (at least one forager is successful for every foraging site).
The other way in which this can happen is that none of the foragers are successful. We thus have the following upper bound:
\begin{equation}\label{P1P2}
\begin{split}
\norm{\bm{q}_{\bm{x}}}_1&\leq \mathbb{P}[V_{i,n}\geq 1 \text{ for all } 1\leq i\leq m | \bm{U}_n=\bm{x}] \\
&+\mathbb{P}[V_{i,n}=0 \text{ for all } 1\leq i\leq m | \bm{U}_n=\bm{x}]. 
\end{split}
\end{equation}
The second probability has already been estimated:
\begin{equation}\label{P2est}
\mathbb{P}[V_{i,n}=0 \text{ for all } 1\leq i\leq m | \bm{U}_n=\bm{x}]\leq \alpha_*^{\xi}
\end{equation}
For the first probability, we have:
\begin{equation}
\begin{split}
&\mathbb{P}[V_{i,n}\geq 1 \text{ for all } 1\leq i\leq m | \bm{U}_n=\bm{x}]=\prod_{i=1}^m \mathbb{P}[V_{i,n}\geq 1| U_{i,n}=x_i]\\
=&\prod_{i=1}^m \paren{1-\mathbb{P}[V_{i,n}=0| U_{i,n}=x_i]}=\prod_{i=1}^m \paren{1-(1-\phi_i(x_i/\xi))^{x_i}}\\
=&\prod_{i=1}^m \paren{1-\paren{\frac{x_i/\xi}{(x_i/\xi) +K_i}}^{x_i}}
\end{split}
\end{equation}
Let us now rewrite the above expression using $u_i=x_i/\xi$. We have:
\begin{equation}\label{P1est}
\mathbb{P}[V_{i,n}\geq 1 \text{ for all } 1\leq i\leq m | \bm{U}_n=\bm{x}]\leq \prod_{i=1}^m \paren{1-\paren{\frac{u_i}{u_i +K_i}}^{u_i\xi}}
\end{equation}
Note that the $u_i$ satisfy:
\begin{equation}\label{simplex}
\sum_{i=1}^m u_i=1, \; 0\leq u_i\leq 1.
\end{equation}
Define the vector $\bm{u}=(u_1,\cdots, u_m)^{\rm T}$. Then the vector $\bm{u}$ lies in the above $m-1$ dimensional simplex $\Sigma$.
Define:
\begin{equation}
G_i(\bm{u})=\paren{\frac{u_i}{u_i+K_i}}^{u_i}, \; i=1,\cdots m.
\end{equation}
The above is only defined for $u_i>0$. However, 
since $\lim_{z\to 0+} z^z=1$, we can extend the above functions to be continuous functions on $\Sigma$ (including the edges).
Inequalities \eqref{P1est}, \eqref{P2est} and \eqref{P1P2} thus yield:
\begin{equation}\label{qxest1}
\norm{\bm{q}_{\bm{x}}}_1\leq \prod_{i=1}^m \paren{1-G_i(\bm{u})^\xi}+\alpha_*^{\xi}, \; \bm{x}=\xi\bm{u}.
\end{equation}
Noting that $0\leq G_i(\bm{u})\leq 1$, we have:
\begin{equation}\label{qxest2}
\norm{\bm{q}_{\bm{x}}}_1\leq \min_{1\leq i\leq m}\paren{1-G_i(\bm{u})^\xi}+\alpha_*^\xi=1-\paren{\max_{1\leq i\leq m} G_i(\bm{u})}^\xi +\alpha_*^{\xi}.
\end{equation}
We prove a technical lemma.
\begin{lemma}\label{alphabeta}
Recall that $\Sigma$ is the closed $m-1$ dimensional simplex defined in \eqref{simplex}. We have:
\begin{equation}
\alpha_*<\beta_*=\min_{\bm{u}\in \Sigma} \max_{1\leq i\leq m} G_i(\bm{u})\leq 1,
\end{equation}
where $\alpha_*$ is defined in \eqref{alphastar}.
\end{lemma}
\begin{proof}
Let 
\begin{equation}
G_{\rm max}(\bm{u})=\max_{1\leq i\leq m} G_i(\bm{u}).
\end{equation}
The function $G_{\rm max}(\bm{u})$ is continuous on $\Sigma$ since the functions $G_i(\bm{u})$ are continuous on $\Sigma$.
Note that:
\begin{equation}
\PD{}{u_i} \log G_i(\bm{u})=\log\paren{\frac{u_i}{u_i+K_i}}-\paren{\frac{u_i}{u_i+K_i}-1}<0
\end{equation}
where we used the fact that $\log(y)-(y-1)<0$ for $y<0$. This means that $G_i(\bm{u})$, which only depends on $u_i$, 
is a strictly monotone decreasing function of $u_i$. Thus, for $\bm{u}\in \Sigma$, 
\begin{equation}
G_i(\bm{u})=\frac{1}{1+K_i} \text{ when } u_i=1,\quad G_i(\bm{u})>\frac{1}{1+K_i} \text{ otherwise}. 
\end{equation}
In particular, 
\begin{equation}
G_i(\bm{u})> \frac{1}{1+K_i} \text{ for } \bm{u}\in \Sigma \backslash \partial \Sigma,
\end{equation}
where $\partial \Sigma$ is the boundary of $\Sigma$. This shows that
\begin{equation}
G_{\rm max}(\bm{u})> \max_{1\leq i\leq m}\frac{1}{1+K_i}=\alpha_* \text{ for } \bm{u}\in \Sigma \backslash \partial \Sigma.
\end{equation}
Let us examine the value of $G_{\rm max}(\bm{u})$ for $\bm{u}\in \partial \Sigma$. On $\partial \Sigma$, at least one of the coordinates $u_i$ is equal to $0$.
Note that $G_i(0)=1$. Thus, 
\begin{equation}
G_{\rm max}(\bm{u})=1>\alpha_* \text{ for } \bm{u}\in \partial \Sigma.
\end{equation}
We thus see that:
\begin{equation}
G_{\rm max}(\bm{u})>\alpha_* \text{ for } \bm{u}\in \Sigma.
\end{equation}
Since $G_{\rm max}$ is continuous and $\Sigma$ is compact, it attains a minimum at some point $\bm{u}_*\in \Sigma$.
Thus,
\begin{equation}
\beta_*=\min_{\bm{u}\in \Sigma} G_{\rm max}(\bm{u})=G_{\rm max}(\bm{u}_*)>\alpha_*.
\end{equation}
\end{proof}
The above lemma, together with \eqref{qxest1}, shows that:
\begin{equation}
\norm{\bm{q}_{\bm{x}}}_1\leq 1-\beta_*^\xi +\alpha_*^{\xi}, \; \alpha_*<\beta_*\leq 1.
\end{equation}
So, we have:
\begin{equation}\label{ABBest}
\norm{A_{\mc{B}\mc{B}}}_1=\sup_{\bm{x}\in \mc{B}} \norm{\bm{q}_{\bm{x}}}_1 \leq 1-\beta_*^\xi +\alpha_*^{\xi}.
\end{equation}
It is now a simple matter to prove the following proposition.
\begin{proposition}\label{p:piB}
Consider the $\bm{\pi}_{\mc{B}}$ in \eqref{ABE}. We have:
\begin{equation}
\norm{\bm{\pi}_{\mc{B}}}_1\leq 2\paren{\frac{\alpha_*}{\beta_*}}^{\xi}.
\end{equation}
In particular, 
\begin{equation}
\lim_{\xi \to \infty} \norm{\bm{\pi}_{\mc{B}}}_1=0.
\end{equation}
\end{proposition}
\begin{proof}
For the first assertion, first note that:
\begin{equation}
\beta_*^\xi -\alpha_*^{\xi}\geq \frac{1}{2}\beta_*^\xi \text{ for } \xi\geq \xi_0=\log_{\beta_*/\alpha_*} 2.
\end{equation}
This is possible since $\beta_*>\alpha_*$. Using \eqref{ABBest} together with the above, we thus have:
\begin{equation}
\norm{A_{\mc{B}\mc{B}}}_1\leq 1-\frac{1}{2}\beta_*^\xi \text{ for } \xi\geq \xi_0.
\end{equation}
Using this, together with \eqref{ABEest} and \eqref{piBest}, we have:
\begin{equation}
\norm{\bm{\pi}_{\mc{B}}}_1\leq 2\paren{\frac{\alpha_*}{\beta_*}}^{\xi} \text{ for } \xi\geq \xi_0.
\end{equation}
Since $\norm{\bm{\pi}_{\mc{B}}}_1\leq 1$, the above holds even for $\xi<\xi_0$.
The second assertion follows from the first since $\alpha_*<\beta_*$ by Lemma \ref{alphabeta}.
\end{proof}

\subsection{Existence of intermediate optimal recruitment} \label{sect4.3}
The above has the following implication for the long-term reward rate $\mu_\xi(\rho)$ and $\mu(\rho)$ near $\rho = 1$. The reward rate satisfies
\begin{equation}
\mu_\xi(\rho)=\sum_{\bm{x}\in \mc{S}} \paren{\sum_{i=1}^m\paren{f_i(x_i/\xi)}}\pi({\bm{x}}).
\end{equation}
Now, let us consider the case $\rho=1$. We split the sum:
\begin{equation}\label{winfty}
\mu_\xi(1)=\sum_{\bm{x}\in \mc{B}} \paren{\sum_{i=1}^m\paren{f_i(x_i/\xi)}}\pi({\bm{x}})+\sum_{\bm{x}\in \mc{E}} \paren{\sum_{i=1}^m\paren{f_i(x_i/\xi)}}\pi({\bm{x}})
\end{equation}
Define:
\begin{equation}
\mc{E}_j=\lbrace \bm{\bm{x}}=(x_1,\cdots,x_m)^{\rm T}\in \mc{E}| x_j=0\rbrace. 
\end{equation}
Clearly, 
\begin{align}
\sum_{\bm{x}\in \mc{E}} \paren{\sum_{i=1}^m\paren{f_i(x_i/\xi)}}\pi({\bm{x}}) &\leq \max_{\bm{x}\in \mc{E}}\paren{\sum_{i=1}^m\paren{f_i(x_i/\xi)}} \nonumber \\
&\leq \max_{1\leq j\leq m}\max_{\bm{x}\in \mc{E}_j}\paren{\sum_{i=1}^m\paren{f_i(x_i/\xi)}}.
\end{align}
From the calculation in \eqref{wmaximization}, we know that
\begin{equation}
\max_{\bm{x}\in \mc{E}_j}\paren{\sum_{i=1}^m\paren{f_i(x_i/\xi)}}=\frac{\sum_{\ell\neq j} K_\ell}{1+\sum_{\ell \neq j} K_\ell}.
\end{equation}
Without loss of generality, let us order the $K_i$ so that
\begin{equation}
K_1\geq K_2\geq \cdots \geq K_m>0.
\end{equation}
Then, combining the above chain of inequalities, we have:
\begin{equation}\label{winf2est}
\sum_{\bm{x}\in \mc{E}} \paren{\sum_{i=1}^m\paren{f_i(x_i/\xi)}}\pi({\bm{x}})\leq \max_{1\leq j\leq m}\frac{\sum_{\ell\neq j} K_\ell}{1+\sum_{\ell \neq j} K_\ell}=\frac{\sum_{\ell=1}^{m-1} K_\ell}{1+\sum_{\ell=1}^{m-1} K_\ell}
\end{equation}
For the other term in \eqref{winfty}, we have:
\begin{equation}\label{winf1est}
\begin{split}
\sum_{\bm{x}\in \mc{B}} \paren{\sum_{i=1}^m\paren{f_i(x_i/\xi)}}\pi({\bm{x}})&\leq \max_{\bm{x}\in \mc{B}}\paren{\sum_{i=1}^m\paren{f_i(x_i/\xi)}}\norm{\bm{\pi}_{\mc{B}}}_1\\
&=\frac{2\sum_{\ell=1}^m K_\ell}{1+\sum_{\ell=1}^m K_\ell} \paren{\frac{\alpha_*}{\beta_*}}^\xi.
\end{split}
\end{equation}
Combining \eqref{winf1est} and \eqref{winf2est} with \eqref{winfty}, we have:
\begin{equation}
\mu_\xi(1)\leq \frac{2\sum_{\ell=1}^m K_\ell}{1+\sum_{\ell=1}^m K_\ell}\paren{\frac{\alpha_*}{\beta_*}}^{\xi}+\frac{\sum_{\ell=1}^{m-1} K_\ell}{1+\sum_{\ell=1}^{m-1} K_\ell} 
\end{equation}
Therefore, for sufficiently large $\xi$, we always have:
\begin{equation}
\mu_\xi(1)<\frac{2\sum_{\ell=1}^m K_\ell}{1+\sum_{\ell=1}^m K_\ell}\paren{\frac{\alpha_*}{\beta_*}}^{\xi}+\frac{\sum_{\ell=1}^{m-1} K_\ell}{1+\sum_{\ell=1}^{m-1} K_\ell} < \frac{\sum_{\ell=1}^m K_\ell}{1+\sum_{\ell=1}^m K_\ell}.
\end{equation}
The right-hand side of the above inequality is what we expect from the deterministic model. For every finite $\xi$, $\mu_\xi(\rho)$ is continuous in $\rho$ up to $\rho=1$.
Thus, at and near $\rho=1$, the total reward from the stochastic model is always less than that from the deterministic model for sufficiently large $\xi$
\begin{equation} \label{gap}
    \lim_{\xi \to \infty} \lim_{\rho \to 1^-}\mu_\xi(\rho) < \lim_{\rho \to 1^-}\lim_{\xi \to \infty} \mu_\xi(\rho).
\end{equation}
In other words, those limiting operators are not interchangable at $\rho = 1$.

On the other hand, if $0 \leq \rho < 1$, 
the stationary distribution of the stochastic model is expected to concentrate near the unique fixed point of the deterministic model, as we discussed in \eqref{poplim_pdf}.
This implies that
\[
    \lim_{\xi \to \infty} \mu_\xi(\rho) = \mu(\rho), \quad 0 \leq \rho < 1.
\]
Then there exists a boundary layer at $\rho = 1$, as seen in Fig. \ref{fig4}, which follows $\rho^\star < 1$.
Similarly, we expect that there is a boundary layer for the model variation in \eqref{model2}. That is, the boundary layer does not come from a peculiar model choice.

\section{Discussion}

We have developed a model of collective foraging from a central site, for both an infinite and a finite population. Even though the central recurrence relationship for the fraction of foragers at each site converges to the infinite-population case by the law of large numbers, the behavior of the finite-population model is qualitatively difference from its infinite-population analogue, regardless of the population size. In particular, the long-term reward rate is maximized when successful individuals always recruit others to their known site, in the infinite population case; but for any finite population, the reward rate is maximized by intermediate levels of communication and recruitment.

There are many open questions and avenues for future research based on the framework we have developed.
One crucial set of questions is how foraging efficiency changes with foragers social structures. For example, within a population one can consider a finer communication network structure among foragers \cite{Madhushani2021a} instead of mass action.  Introducing a subgroup of specialized foragers \cite{Staps2022}, such as dedicated searchers who never follow other recruiters, may improve the efficiency at high recruitment rates in the remainder of foragers. Another related question concerns the optimal recruitment rate in the context of multiple competing groups (hives), when foraging at a site is subject to both within-group competition as well as between group competition. Finally, adapting to changing environment is another critical factor for the survival of social foragers \cite{Kilpatrick2021,Barendregt2022,Staps2022,Arehart2022}. For example, what is the optimal recruitment probability in the presence of dynamic resource availability $K_l$? 
All of these remain as interesting open questions that our model may be generalized to study.

Although the existence of a boundary layer in our model is strongly suggested by Eq. \eqref{gap}, it still remains to prove the convergence of the stationary distribution \eqref{poplim_pdf}. This might be achieved by showing the global stability of the unique fixed point of the population limit model \eqref{veceqn}.
Also, this boundary layer does not arise from a peculiar choice of foraging probability $\phi_i(u_i)$ in \eqref{success_prob} and \eqref{veqn}. In fact, it arises from the communication structure (represented by recruitment probability $\rho$) in our model, which makes multiple fixed points for \eqref{veceqn} only at $\rho = 1$. A boundary layer can exist at $\rho = 1$ as long as $\phi_i(u)$ is decreasing because the foraging site with highest $\phi_i(1)$ is preferred, as seen in Sect. \ref{sect4.1}.

\section*{Acknowledgments} This work was supported by the Simons Foundation, USA (Math+X Grant to Y.M.) and the National Science Foundation, USA (Grant No. DMS-2042144 to Y.M.).

\bibliographystyle{siamplain}
\bibliography{references}
\end{document}